\documentclass[twoside,leqno,twocolumn]{article}

\usepackage[utf8]{inputenc}

\usepackage[english]{babel}

\usepackage{graphicx}
\usepackage{booktabs}
\usepackage{microtype}
\usepackage{xspace}
\usepackage{ifthen}
\usepackage{tikz}
\usepackage{pgfplots}
\usepackage{siunitx}
\usepackage{subfig}
\usepackage[ruled, vlined, linesnumbered, nofillcomment]{algorithm2e}
\usepackage{url}

\usepackage[bm,nocitepackage]{pauli_new} 

\usepackage[numbers,sort&compress,sectionbib]{natbib}

\usepackage[theoremfont]{newtxtext}
\usepackage[varg]{newtxmath}

\usepackage{hyperref}

\renewcommand{\vec}[1]{\bm{#1}}

\graphicspath{{figs/}}

\renewcommand{\rem}[1]{\marginpar{\footnotesize\raggedright #1}}

\newcommand{\set}[1]{\left\{#1\right\}}
\newcommand{\pr}[1]{\left(#1\right)}
\newcommand{\fpr}[1]{\mathopen{}\left(#1\right)}

\newcommand{\np}{\textbf{NP}}

\newcommand{\define}{\leftarrow}

\DeclareRobustCommand{\dispfunc}[2]{%
  \ensuremath{%
  \ifthenelse{\equal{#2}{}}%
    {\mathit{#1}}%
    {\mathit{#1}\fpr{#2}}}}

\pgfdeclarelayer{background}
\pgfdeclarelayer{foreground}
\pgfsetlayers{background,main,foreground}

\definecolor{yafaxiscolor}{rgb}{0.3, 0.3, 0.3}

\definecolor{yafcolor1}{rgb}{0.4, 0.165, 0.553}
\definecolor{yafcolor2}{rgb}{0.949, 0.482, 0.216}
\definecolor{yafcolor3}{rgb}{0.47, 0.549, 0.306}
\definecolor{yafcolor4}{rgb}{0.925, 0.165, 0.224}
\definecolor{yafcolor5}{rgb}{0.141, 0.345, 0.643}
\definecolor{yafcolor6}{rgb}{0.965, 0.933, 0.267}
\definecolor{yafcolor7}{rgb}{0.627, 0.118, 0.165}
\definecolor{yafcolor8}{rgb}{0.878, 0.475, 0.686}

\newlength{\yafaxispad}
\newlength{\yaftlpad}
\newlength{\yaflabelpad}
\newlength{\yafaxiswidth}
\newlength{\yafticklen}
\makeatletter
\def\pgfplots@drawtickgridlines@INSTALLCLIP@onorientedsurf#1{}
\makeatother

\newcommand{\yafdrawxaxis}[2]{
	\pgfplotstransformcoordinatex{#1}\let\xmincoord=\pgfmathresult 
	\pgfplotstransformcoordinatex{#2}\let\xmaxcoord=\pgfmathresult 
	\pgfsetlinewidth{\yafaxiswidth} 
	\pgfsetcolor{yafaxiscolor}
	\pgfpathmoveto{\pgfpointadd{\pgfpointadd{\pgfplotspointrelaxisxy{0}{0}}{\pgfqpointxy{\xmincoord}{0}}}{\pgfqpoint{-0.5\yafaxiswidth}{\yafaxispad}}}
	\pgfpathlineto{\pgfpointadd{\pgfpointadd{\pgfplotspointrelaxisxy{0}{0}}{\pgfqpointxy{\xmaxcoord}{0}}}{\pgfqpoint{0.5\yafaxiswidth}{\yafaxispad}}}
	\pgfusepath{stroke}

}
\newcommand{\yafdrawyaxis}[2]{
	\pgfplotstransformcoordinatey{#1}\let\ymincoord=\pgfmathresult 
	\pgfplotstransformcoordinatey{#2}\let\ymaxcoord=\pgfmathresult 
	\pgfsetlinewidth{\yafaxiswidth} 
	\pgfsetcolor{yafaxiscolor}
	\pgfpathmoveto{\pgfpointadd{\pgfpointadd{\pgfplotspointrelaxisxy{0}{0}}{\pgfqpointxy{0}{\ymincoord}}}{\pgfqpoint{\yafaxispad}{-0.5\yafaxiswidth}}}
	\pgfpathlineto{\pgfpointadd{\pgfpointadd{\pgfplotspointrelaxisxy{0}{0}}{\pgfqpointxy{0}{\ymaxcoord}}}{\pgfqpoint{\yafaxispad}{0.5\yafaxiswidth}}}
	\pgfusepath{stroke}
}

\newcommand{\yafdrawaxis}[4]{\yafdrawxaxis{#1}{#2}\yafdrawyaxis{#3}{#4}}

\pgfplotscreateplotcyclelist{yaf}{%
{yafcolor1,mark options={scale=0.75},mark=o}, 
{yafcolor2,mark options={scale=0.75},mark=square},
{yafcolor3,mark options={scale=0.75},mark=triangle},
{yafcolor4,mark options={scale=0.75},mark=o},
{yafcolor5,mark options={scale=0.75},mark=o},
{yafcolor6,mark options={scale=0.75},mark=o},
{yafcolor7,mark options={scale=0.75},mark=o},
{yafcolor8,mark options={scale=0.75},mark=o}} 

\pgfplotsset{axis y line=left, axis x line=bottom,
	tick align=outside,
	compat = 1.3,
	tickwidth=\yafticklen,
	clip = false,
	every axis title shift = 0pt,
    x axis line style= {-, line width = 0pt, opacity = 0},
    y axis line style= {-, line width = 0pt, opacity = 0},
    x tick style= {line width = \yafaxiswidth, color=yafaxiscolor, yshift = \yafaxispad},
    y tick style= {line width = \yafaxiswidth, color=yafaxiscolor, xshift = \yafaxispad},
    x tick label style = {font=\scriptsize, yshift = \yaftlpad},
    y tick label style = {font=\scriptsize, xshift = \yaftlpad},
    every axis y label/.style = {at = {(ticklabel cs:0.5)}, rotate=90, anchor=center, font=\scriptsize, yshift = -\yaflabelpad},
    every axis x label/.style = {at = {(ticklabel cs:0.5)}, anchor=center, font=\scriptsize, yshift = \yaflabelpad},
    x tick label style = {font=\scriptsize, yshift = 1pt},
    grid = major,
    major grid style  = {dash pattern = on 1pt off 3 pt},
	every axis plot post/.append style= {line width=\yafaxiswidth} ,
	legend cell align = left,
	legend style = {inner sep = 1pt, cells = {font=\scriptsize}},
	legend image code/.code={%
		\draw[mark repeat=2,mark phase=2,#1] 
		plot coordinates { (0cm,0cm) (0.15cm,0cm) (0.3cm,0cm) };%
	} 
}

\newcommand{\bsum}{\lor}

\newcommand{\order}[1]{\dispfunc{order}{#1}}
\newcommand{\inner}[1]{\dispfunc{inner}{#1}}
\newcommand{\border}[1]{\dispfunc{border}{#1}}
\newcommand{\total}[1]{\dispfunc{total}{#1}}
\newcommand{\children}[1]{\dispfunc{ch}{#1}}
\newcommand{\bigO}[1]{\mathcal{O}\fpr{#1}}

\newcommand{\probname}[1]{\textsc{#1}}
\newcommand{\obmf}{\probname{obmf}\xspace}
\newcommand{\obmfstep}{\probname{obmfstep}\xspace}
\newcommand{\cobmfstep}{\probname{cobmfstep}\xspace}
\newcommand{\optset}{\probname{optset}\xspace}
\newcommand{\obmfsym}{\probname{obmf}\textsubscript{\textup{\textrm{sym}}}\xspace}
\newcommand{\cobmf}{\probname{cobmf}\xspace}
\newcommand{\cobmfsym}{\probname{cobmf}\textsubscript{\textup{\textrm{sym}}}\xspace}

\newcommand{\hamilton}{\probname{Hamilton path}\xspace}

\newcommand{\aobmf}{\algname{obmf}\xspace}
\newcommand{\acobmf}{\algname{cobmf}\xspace}
\newcommand{\aobmfsym}{\algname{obmf}\textsubscript{sym}\xspace}
\newcommand{\acobmfsym}{\algname{cobmf}\textsubscript{sym}\xspace}
\newcommand{\asso}{\algname{asso}\xspace}
\newcommand{\assosym}{\algname{asso}\textsubscript{sym}\xspace}

\renewcommand{\datasetname}[1]{\textit{#1}}
\newcommand{\lesmis}{\datasetname{Les Mis\'erables}\xspace}
\newcommand{\now}{\datasetname{Paleo}\xspace}
\newcommand{\news}{\datasetname{Newsgroups}\xspace}
\newcommand{\newssym}{\datasetname{Terms}\xspace}
\newcommand{\locations}{\datasetname{Locations}\xspace}
\newcommand{\mammals}{\datasetname{Mammals}\xspace}
\newcommand{\lesmisa}{\datasetname{Les Mis}\xspace}
\newcommand{\newsa}{\datasetname{News}\xspace}

\newcommand{\paperurl}{\url{https://cs.uef.fi/~pauli/bmf/ordered_bmf/}}

\begin{document}
\title{Boolean matrix factorization meets consecutive ones property\thanks{This is an extended version of the paper of the same name presented in 2019 SIAM International Conference on Data Mining.}}

\author{%
  Nikolaj Tatti\thanks{University of Helsinki, Helsinki, Finland,\newline \texttt{nikolaj.tatti@helsinki.fi}}
    \and
  Pauli Miettinen\thanks{University of Eastern Finland, Kuopio, Finland,\newline \texttt{pauli.miettinen@uef.fi}. Part of this work was done while the author was with MPI-INF, Saarbr\"ucken, Germany.}
}

\date{}

\maketitle              


\begin{abstract}
  \small\baselineskip=9pt
  Boolean matrix factorization is a natural and a popular technique for
summarizing binary matrices. 
In this paper, we study a problem of Boolean matrix factorization where we
additionally require that the factor matrices have consecutive ones property (OBMF).
A major application of this optimization problem comes from graph visualization:
standard techniques for visualizing graphs are circular or linear layout, where
nodes are ordered in circle or on a line. A common problem with visualizing
graphs is clutter due to too many edges. The standard approach to deal with
this is to bundle edges together and represent them as ribbon. We also show that we
can use OBMF for edge bundling combined with circular or linear layout
techniques. 

We demonstrate that not only this problem is \NP-hard but we cannot have a
polynomial-time algorithm that yields a multiplicative approximation guarantee
(unless $\Poly = \NP$).  On the positive side, we develop a greedy algorithm
where at each step we look for the best 1-rank factorization. Since even
obtaining 1-rank factorization is \NP-hard, we propose an iterative algorithm
where we fix one side and and find the other, reverse the roles, and repeat.
We show that this step can be done in linear time using pq-trees.
We also extend the problem to cyclic ones property and symmetric factorizations. Our experiments show that our algorithms find high-quality factorizations and scale well.

\end{abstract}


\section{Introduction}
\label{sec:introduction}

Matrix factorization is an immensely popular way of summarizing data as well as
discovering signal from the data. While being useful, the interpretation and
visualization of discovered factor matrices may be difficult. A popular variant
for factorizing binary matrices is a $k$-Boolean matrix factorization, which,
essentially, summarizes the binary data as a union of $k$ tiles, that is, submatrices
full of 1s. However, visualizing such factorization is difficult as the
discovered rows and columns can be any sets, and there is no insightful way 
of visualizing them all at once.

In this paper we consider $k$-Boolean matrix factorization such that the
resulting matrix has a certain property: we can order the columns and the rows
such that the matrix consists of union of $k$ \emph{contiguous} tiles.  We do
not know the order before-hand, and we discover the order as we also discover
the factorization.

Our motivation for discovering such factorization is primarily due to easy
exploration of the factorization: we can draw the factorization as $k$ tiles.
While in certain cases, such a constraint may be too restrictive, there are
many settings, where this constraint comes naturally. As a specific example,
consider visualizing graphs. A classic technique for visualizing a graph is
using linear or circular layout, where the nodes are drawn on a line or circle,
and they are connected with arcs. The most common problem with visualizing
graphs is clutter due to too many edges. To combat the clutter, edges are often
grouped, and drawn in ribbons (see Figure~\ref{fig:lesmis} for an example). The problem is to discover such ribbons and the node order, while minimizing
the error. We show that we can use matrix factorization on the adjacency matrix
of a graph to find the order and the groups.

We show that the factorization we seek can be expressed with consecutive
ones property (C1P).  Namely, we will look for factor matrices $\mX$ and $\mY$
whose columns can be shuffled such that each row has a form of $[0, \ldots, 0,
1, \ldots, 1, 0, \ldots, 0]$. We show that the problem is \np-hard, even if $k
= 1$, and it is inapproximable for $k > 1$. On the positive side, we propose a
greedy algorithm that searches the factors in iterative manner.
The search is done by first fixing a vector in $\mX$ and finding the optimal
counterpart in $\mY$, then fixing the vector in $\mY$ and finding the optimal
vector in $\mX$, and so on, until convergence. We show that we can find the optimal
counterpart in linear time using pq-trees.

We also consider 3 extensions of this factorization: the first variant, cyclic
decomposition, consists of allowing factors to ``wrap around the border.''  the
second variant is specifically designed for symmetric matrices, while the last
variant combines the two. Performing cyclic and symmetric decomposition proves
to be useful for cyclic layout of graphs.

The rest of the paper is organized as follows: We present preliminary notation
and define the matrix factorization and the cyclic version in Section~\ref{sec:notation-definitions}.
We present the search algorithm in Section~\ref{sec:algorithm}.
The symmetric extensions are given in Section~\ref{sec:extensions}.
Section~\ref{sec:related-work} is dedicated to related work, and Section~\ref{sec:experiments} is dedicated to experimental evaluation.
Finally, we conclude the paper with remarks given in Section~\ref{sec:conclusions}.
All proofs are given in Appendix~\ref{sec:appendix:proofs}.


\section{Preliminary notation and problem definitions}
\label{sec:notation-definitions}
We begin by presenting preliminary notation, and then present the two main problem definitions.
Extended problems are discussed in Section~\ref{sec:extensions}.

\subsection{Notation}
\label{sec:notation}
Given an \by{n}{k} binary matrix $\mA$ and a \by{k}{m} binary matrix $\mB$, the \emph{Boolean matrix product} $\mA\bprod\mB$ is defined element-wise as
\begin{equation}
  \label{eq:bprod}
  (\mA\bprod\mB)_{ij} = \bigvee_{\ell = 1}^k a_{i\ell}b_{\ell j}\; .
\end{equation}
The \emph{Boolean matrix sum} of $\mA\in\B^{n\times m}$ and $\mB\in\B^{n\times m}$ is defined elementwise as $(\mA\bsum\mB)_{ij} = a_{ij}\lor b_{ij}$.

To measure the distance between two binary matrices, we use the \emph{squared Frobenius norm} of their (normal) difference, $\norm{\mA-\mB}_F^2$. Notice that as $\mA$ and $\mB$ are both binary, this is the same as calculating the number of disagreements between $\mA$ and $\mB$: $\norm{\mA-\mB}_F^2 = \abs{\{(i,j) : a_{ij} \neq b_{ij}\}}$.

We say that a binary matrix $\mX$ has a \emph{consecutive ones property} (C1P)
if its columns can be permuted such that each row has a form of $[0, \ldots, 0,
1, \ldots, 1, 0, \ldots, 0]$, that is, 1s form a contiguous interval.
For the sake of presentation, we will also refer these matrices as \emph{unimodal}.

We say that a binary matrix $\mX$ is \emph{cyclic} if its columns can be
permuted such that each row has a form of $[0, \ldots, 0, 1, \ldots, 1, 0,
\ldots, 0]$ or $[1, \ldots, 1, 0, \ldots, 0, 1, \ldots, 1]$.

\subsection{Problem definitions}
\label{sec:problem-definitions}
Next we will give our two main optimization problems.

\begin{problem}[Ordered BMF, \obmf]
  \label{prob:obmf}
  Given a binary matrix $\mD$ and an integer $k\in\N$, find 
  two unimodal binary
  matrices $\mX$ and $\mY$ that minimize the number of disagreements
  \begin{equation}
    \label{eq:obmf}
    \norm*[big]{\mD - (\mX^T\bprod\mY)}_F^2\; .
  \end{equation}
\end{problem}

\begin{problem}[Cyclic Ordered BMF, \cobmf]
  \label{prob:cobmf}
  Given a binary matrix $\mD$ and an integer $k\in\N$, find 
  two cyclic binary
  matrices $\mX$ and $\mY$ that minimize the number of disagreements
  \begin{equation}
    \label{eq:cobmf}
    \norm*[big]{\mD - (\mX^T\bprod\mY)}_F^2\; .
  \end{equation}
\end{problem}

The matrix $\mZ = \mX^T\bprod\mY$ given in Eq.~\ref{eq:obmf} has another natural alternative
characterization: the columns and the rows of $\mZ$ can be permuted such that the resulting
matrix is a union of $k$ contiguous tiles of 1s. Similarly, the matrix $\mZ = \mX^T\bprod\mY$ given in Eq~\ref{eq:cobmf}
can be permuted such that the resulting matrix is a union of $k$ contiguous tiles, but we also allow the tiles to wrap around the border.



Unsurprisingly, the problems are computationally infeasible.
First, we demonstrate that \obmf is difficult even if $k = 1$.

\begin{theorem}
  \label{thm:obmf:np_hard}
  The \obmf problem is \NP-hard,  even if $k = 1$.
\end{theorem}

Our next result shows that not only \obmf is difficult, but it is also
impossible to approximate. To show this, it is enough to demonstrate
that testing for zero-error solution is expensive.

\begin{theorem}
\label{thm:obmf:inapprox}
Deciding whether \obmf has a zero-error solution is \NP-complete.
\end{theorem}

The proofs of these and other statements are given in Appendix~\ref{sec:appendix:proofs}.

\section{Iterative greedy algorithm}
\label{sec:algorithm}

\subsection{Greedy algorithm}

As we saw in the previous section, not only the problem is \NP-hard, we cannot
construct any polynomial-time algorithm with a multiplicative guarantee.
Hence, we need to resort to heuristics.  The most natural heuristic is a greedy
heuristic, where given a $(k - 1)$-sized factorization we look for a $k$-sized
factorization by adding one row and one column to $\mX$ and $\mY$.  Note that
these rows need to be selected carefully such that $\mX$ and $\mY$ remain
unimodal, and we also need to maintain the permutation(s).

Unfortunately, Theorem~\ref{thm:obmf:np_hard} states that we cannot even find
the best solution for $k = 1$ in polynomial-time. Fortunately, we can solve
quickly a subproblem, where we have fixed one side.

\begin{problem}[Ordered BMF step, \obmfstep]
\label{prob:obmfstep}
Given a binary matrix $\mD$ of size $\by{n}{m}$ and two unimodal matrices, $\mX'$ of size $\by{k}{n}$
and $\mY'$ of size $\by{(k - 1)}{m}$,
find the decomposition $\mX^T \bprod \mY$
solving \obmf such that $\mX = \mX'$ and $\mY$ is obtained by adding one new row to $\mY'$.
\end{problem}

We can use \obmfstep as follows. Assume that we have already found $\by{(k - 1)}{m}$ matrices $\mX$ and $\mY$.
We first extend $\mX$ with a new row using a given seed, and find
the optimal new row for $\mY$ (strategy for such selection is given later
using \obmfstep. We fix the discovered row, and use \obmfstep to find the corresponding
row for $\mX$.
Since we solve each step optimally,
the error will never increase. We stop when the error stops decreasing.
Note that we will need to provide a seed for the initial row in $\mX$. 
Here, we test several possible seeds $S$, and select the best. We experiment
with several options in experiments, but the default is that $S$ is equal to all singleton columns.
The pseudo-code for the algorithm is given in Algorithm~\ref{alg:greedy}.

\begin{algorithm}[tb]
\caption{Greedy iterative algorithm for estimating \obmf.
The algorithm takes as input the dataset, the desired dimension $k$,
and the seed set $S$ used for selecting the first candidate for a column.  
}
\small
\label{alg:greedy}
$\mX \define$ matrix of size $\by{0}{n}$\;
$\mY \define$ matrix of size $\by{0}{m}$\;
\ForEach{$i = 1 \ldots, k$} {
	\ForEach {$s \in S$} {
		$c \define s$\;
		\While {error decreases} {
			$r \define $ best row for fixed columns $[\mY; c]$\; 
			$c \define $ best column for fixed rows $[\mX; r]$\; 
		}
		$\mX \define [\mX; r]$\;
		$\mY \define [\mY; c]$\;
	}
}

\end{algorithm}

The remainder of this section is about solving \obmfstep in linear time.
Almost the same approach will also work for the cyclic version, \cobmfstep; we
will point the minute difference.

\subsection{Expressing permutations with pq-trees}

The complicated aspect of \obmfstep is that we need to make sure that the new matrix is unimodal.
Luckily, we can use pq-trees,
a classic structure that allows us
to express every permutation for which a set of binary vertices remain
unimodal. In this section we will give a brief review of pq-trees and the
two main properties that are relevant to us.

Assume that we are given a universe $U$; in our case this will be either rows
or columns of the input matrix. A pq-tree is a tree with each leaf
corresponding to $u \in U$. There are two types of non-leaf nodes, these types
will dictate what permutations we can perform on the children. We can permute
children of p-node in any order whereas the order of the children of q-node is
fixed but we can flip the direction. The leaves of the permuted tree will 
then indicate an order.
We will denote such orders by $\order{T}$, where $T$ is the pq-tree.

Two seminal results are important to us.
The first result states that there is a pq-tree $T$ such that $\order{T}$
are exactly the orders under which a set of binary vertices remain
unimodal. 

\begin{theorem}[\citet{booth1976testing}]
Given a universe $U$ and $k$ sets $S_i \subseteq U$, there is
a pq-tree $T$ such that $\order{T}$ are exactly the permutations of $U$
under which each $S_i$ is contiguous.
\end{theorem}

The second result states that we can efficiently update the pq-tree.

\begin{theorem}[\citet{booth1976testing}]
Assume that we have a pq-tree $T$ over a universe $U$ and a set $S \subseteq T$. 
Let $P$ be the set of all permutations of $U$ where $S$ is contiguous.
If $\order{T} \cap P \neq \emptyset$, then
there is an $\bigO{\abs{U}}$-time algorithm that constructs a tree $T'$
such that $\order{T'} = \order{T} \cap P$.
If $\order{T} \cap P = \emptyset$, then the same algorithm detects a failure.
\end{theorem}

The detailed description of the algorithm for updating the pq-tree
can be found in~\citep{booth1976testing}.

\subsection{Finding the optimal row}

In this section we describe the algorithm that solves \obmfstep.  Assume that
we have a pq-tree $T$ representing the permutations of columns in $\mD$
allowed by the previously discovered rows in $\mY'$.
When dealing with pq-trees it is notationally easier to deal with sets rather
than with vectors. Naturally every binary vector $\vy$ can be represented as
a set $S = \set{i : y_i = 1}$.

Let us define $U$ to be the column indices of $\mD$; these are exactly the
leaves of $T$.  We say that a set $S \subseteq U$ is \emph{compatible} with a
pq-tree $T$, if there is an order in $\order{T}$ where $S$ is contiguous.
Obviously, compatible sets $S$ correspond exactly to suitable new rows in $\mY$.

We can express \obmfstep as an instance of the following problem.

\begin{problem}[\optset]
Given a universe $U$, weights $w(u)$ for each $u \in U$, and a pq-tree $T$ over
the universe $U$, find a set $S$ that is compatible with $T$ and maximizes
the total weight $\sum_{u \in S} w(u)$.
\end{problem}

Recall that $u \in U$ corresponds to a column index of $\mD$. Define $w(u)$
to be the gain in the error-function if we were to use $u$ in our new row for
$\mY$. More formally, let $\vx$ be the fixed counterpart in $\mX$ for the new
row in $\mY$.
Let $p$ be the number of ones in $\mD$ at rows $\vx$ and column $u$
that are not yet covered by the previous factors.
Let $n$ be the number of zeros in $\mD$ at rows $\vx$ and column $u$
that are not yet covered by the previous factors. We define $w(u) = p - n$.
Solving \optset with these weights solves \obmfstep.

In order to solve \cobmfstep, we solve \optset using $w(u) = p - n$, as above, yielding a set, say $S_1$.
In addition, we also solve \optset using $w(u) = n - p$, yielding a set, say $S_2$.
Then, we use either $S_1$ or $U \setminus S_2$, whichever yields a better gain.

In order to solve \optset, we need an additional definition:
Let $S$ be a compatible set of a pq-tree $T$.
If there is a permutation in $\order{T}$ with the first or the last element
in $S$, we call $S$ a \emph{border-compatible} set. 

Let $T$ be a pq-tree. To solve \optset
we will compute 3 counters for a node $v$ in $T$, 
namely, $\inner{v}$, $\border{v}$, and $\total{v}$.
The counter $\total{v}$ corresponds to the total weight of leaves under $v$,
while the counter $\inner{v}$ corresponds to the best $S$ that is compatible with
the subtree starting at $v$.
Finally,
$\border{v}$ corresponds to the best $S$ that is border-compatible with
the subtree starting at $v$.

We should stress that, strictly by definition, $\inner{v}$ can represent an empty set,
whereas $\total{v}$ and $\border{v}$ should be never empty, even if they produce
a negative value. Thus, $\inner{v} \geq 0$ but $\border{v}$ and $\total{v}$
can have negative values. Moreover, it is possible that $\border{v}$ represents
every leaf of $v$, in which case, $\border{v} = \total{v}$.

Naturally, we want to compute $\inner{r}$, where $r$ is the root of $T$.
To obtain this value we compute each value iteratively, children first.
We also maintain the lists of the children that were responsible for producing
the optimal value. These lists are clear from the proofs of the following lemmata.
This allows us to extract the optimal $S$.

First, note that computing $\total{v}$ is trivial since $\total{v} = \sum_{c \in \children{v}}\total{c}$.
If $v$ is a leaf-node, then $\border{v} = \total{v}$ and $\inner{v} = \max(0, \total{v})$.

The next two lemmata establish how to compute the counters for q-nodes.

\begin{lemma}
\label{lem:borderq}
Let $v$ be a q-node and let $c_1, \ldots, c_\ell$ be its children. Then
\[
\begin{split}
	\border{v} & = \max(x, y), \quad\text{where} \\
	x & = \max_{i} \border{c_i} +  \sum_{j = 1}^{i - 1} \total{c_j}, \\
	y & = \max_{i} \border{c_i} +  \sum_{j = i + 1}^{\ell} \total{c_j}\quad.
\end{split}
\]
\end{lemma}

\begin{lemma}
\label{lem:innerq}
Let $v$ be a q-node and let $c_1, \ldots, c_\ell$ be its children. Then
\[
\begin{split}
	\inner{v} & = \max(x, y), \quad\text{where} \\
	x & = \max_i \inner{c_i}, \\
	y & = \max_{i < j}  \border{c_i} + \border{c_j} + \sum_{\ell = i + 1}^{j - 1} \total{c_\ell}\ . 
\end{split}
\]
\end{lemma}

Our next step is to compute the counters for p-nodes.
For that we need to define the following helper function:
given a node $v$ we define $g(v) = \border{v} - \max(\total{v}, 0)$.
We will use $g(v)$ in the next two lemmata describing on how to compute 
the counters for p-node.

\begin{lemma}
\label{lem:borderp}
Let $v$ be a p-node and let $c_1, \ldots, c_\ell$ be its children.
Define $b = \max g(c_i)$.
Then
\[
	\border{v} = b + \sum_i \max(\total{c_i}, 0)\quad. 
\]
\end{lemma}

Note that since we require the set responsible for $\border{v}$
be non-empty, it is possible that $\border{v} < 0$. This can happen only if $b < 0$
and every child $w$ of $v$ has $\total{w} < 0$.

\begin{lemma}
\label{lem:innerp}
Let $v$ be a p-node and let $c_1, \ldots, c_\ell$ be its children.
Define $b_1$ and $b_2$ be the top-2 values of $g(c_i)$.
Then
\[
\begin{split}
	\inner{v} & = \max(x, y), \quad\text{where} \\
		x & = \max_i \inner{c_i}, \\
		y & = \max(b_1, 0) + \max(b_2, 0) + \sum_i \max(\total{c_i}, 0). 
\end{split}
\]
\end{lemma}

Note that using these lemmas every counter can be trivially solved in linear
time, except for $\inner{v}$, where $v$ is q-node. To compute $\inner{v}$
in linear time, it is enough if we can solve
\[
	\border{c_j} + \max_{i < j}  \border{c_i} + \sum_{\ell = i + 1}^{j - 1} \total{c_\ell}
\]
in \emph{constant} time for a \emph{fixed} $j$. Luckily, we can rewrite this function as
\[
	\border{c_j} + \pr{\sum_{\ell = 1}^{j - 1} \total{c_\ell}} + \max_{i < j} t(i, j),
\]
where 
\[	
	t(i, j) = \max_{i < j}  \border{c_i} - \sum_{\ell = 1}^{i} \total{c_\ell}\quad.
\]
Let $i(j)$ to be the optimal $i$ for a fixed $j$.
Since
\[
	\max_{i < j} t(i, j) = \max \pr{t(j - 1, j) \max_{i < j - 1} t(i, j)},
\]
we have either $i(j) = i(j - 1)$ or $i(j) = j - 1$.
If we were to test each $j$ consecutively, then
this allows us to compute $i(j)$ in constant time: we simply
compare the solution $i = j - 1$ to the best previous solution $i(j - 1)$.

In summary, each counter of $v$ can be computed in $\bigO{\abs{\children{v}}}$.
Thus we need $\bigO{\ell}$, where $\ell$ is the number of nodes in $T$. Since $\ell \in \bigO{\abs{U}}$,
we can compute the counters in $\bigO{\abs{U}}$ time, where $\abs{U}$ is the number of columns in $\mD$.

When computing the counters we also store which children were responsible for this value.
Once we have computed $\inner{r}$, where $r$ is the root of the tree, we can backtrack
to obtain the optimal $S$. This can be also done in linear time.

Computing the weights $w$ in \optset can be done in $\bigO{p}$ time, where $p$
is the number of 1s in the dataset $\mD$ of size $\by{n}{m}$. Consequently,
\obmfstep can be done in $\bigO{p + n + m}$ time.


\section{Symmetric decomposition}
\label{sec:extensions}

We now propose an extension for symmetric
matrices.

\subsection{Definition}


If $\mD$ is symmetric (e.g.\ an adjacency matrix of an undirected graph), we have the following problem:
\begin{problem}[Symmetric \obmf, \obmfsym]
  \label{prob:obmf-sym}
Given a binary matrix $\mD$ and an integer $k\in\N$, find
two binary
matrices $\mX$ and $\mY$ such that $[\mX; \mY]$ is unimodal, that minimize the number of disagreements
  \begin{equation}
    \label{eq:obmf-sym}
    \norm*[big]{\mD - \bigl((\mX^T\bprod\mY) \bsum (\mY^T\bprod\mX)\bigr)}_F^2\; .
  \end{equation}
\end{problem}
We define similarly \cobmfsym, a cyclic and symmetric variant of \obmf.

The unimodality condition in \obmfsym states that we should be able to permute $\mX$ and $\mY$
with the \emph{same} permutation so that the rows are in form of $[0, \ldots, 0,
1, \ldots, 1, 0, \ldots, 0]$.

Notice that we do not use the more common symmetric decomposition
$\mD\approx\mX^T\bprod\mX$ as this would lead to necessarily having the blocks
around the diagonal. 

\subsection{Algorithm}
\label{sec:symm:algorithm}

The discovery algorithm for symmetric \obmf is similar. Like with the regular
\obmf, we use a greedy algorithm as an iterative step for discovering new rows.

The first difference is that we maintain only one pq-tree, corresponding to the
rows in both $\mX$ and $\mY$.

The second difference is that -- as $\mX^T\bprod\mY$ and
$\mY^T\bprod\mX$ can have overlapping 1s -- maximizing \optset does not
necessarily produce the optimal row. Instead, we can show that solving
\optset, with the weights as described in the previous section, minimizes
$\norm*[big]{\mD - \mX^T\bprod\mY}_F^2 + \norm*[big]{\mD -
\mY^T\bprod\mX}_F^2$. It follows easily that minimizing this function
yields a 2-approximation for finding optimal counterpart row.


\section{Experimental evaluation}
\label{sec:experiments}

In this section we study how well the algorithms from Sections~\ref{sec:algorithm} and \ref{sec:symm:algorithm} work with synthetic and real-world data. We denote the algorithms with the same names as the problems they are solving, and differentiate the algorithms from the problems via the font. That is, \aobmf is the algorithm for \obmf, and so on. The algorithms are implemented in C++, and we make the source code and synthetic experiments freely available.\!\footnote{\paperurl}

\subsection{Resilience to Noise}
\label{sec:exp:sym}

We start by evaluating the algorithms' resilience to noise. To that end, we
synthesized random matrices of size $95 \times 95$ with block structure (6 blocks of size $20\times20$ along the diagonal, with 5 overlapping rows and columns) and corrupted those matrices
with flipping a varying amounts of entries. The amount of flipped entries varied from \SIrange{0}{50}{\percent}
(of total elements) and we compared the quality of the results to both the
noise-free matrix and noisy matrix. The results are shown in
Figure~\ref{fig:synth_data}. 

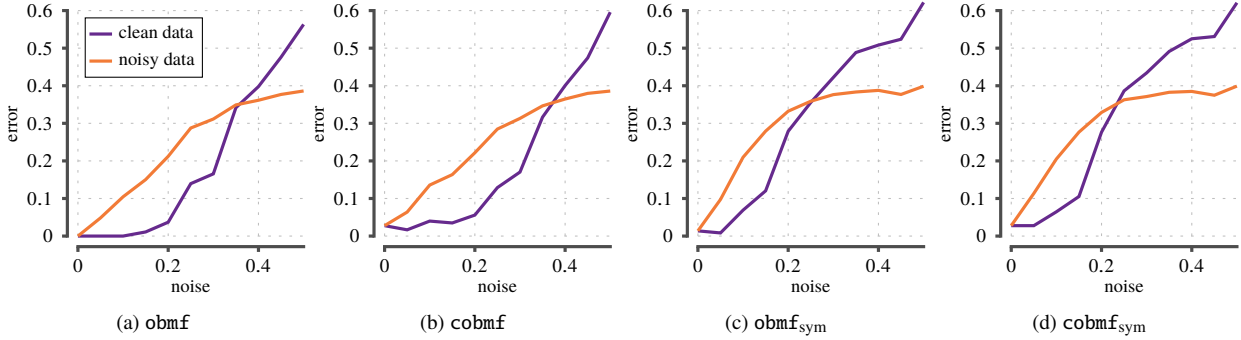
\begin{figure*}[tb!]
  \begin{center}
\pgfplotsset{
    xlabel={noise},ylabel= {error},
	width = 3cm,
    height = 3cm,
    cycle list name=yaf,
    scale only axis,
    tick scale binop=\times,
    x tick label style = {/pgf/number format/set thousands separator = {\,}},
    y tick label style = {/pgf/number format/set thousands separator = {\,}},
    scaled ticks = false,
    xmax = 0.5,
    ymax = 0.6,
	ymin = 0.0,
    ytick = {0, 0.1, ..., 0.7},
    legend pos = north west,
    no markers}
\subfloat[\aobmf]{%
\begin{tikzpicture}
\begin{axis}
\addplot table[x index = 0, y expr = {\thisrowno{1} / (95 * 95)}, header = false] {figs/new_block_error.dat};
\addplot table[x index = 0, y expr = {\thisrowno{2} / (95 * 95)}, header = false] {figs/new_block_error.dat};
\legend {clean data, noisy data};
\pgfplotsextra{\yafdrawaxis{0}{0.5}{0}{0.6}}
\end{axis}
\end{tikzpicture}%
}%
\subfloat[\acobmf]{%
\begin{tikzpicture}
\begin{axis}
\addplot table[x index = 0, y expr = {\thisrowno{1} / (95 * 95)}, header = false] {figs/cyc_block_error.dat};
\addplot table[x index = 0, y expr = {\thisrowno{2} / (95 * 95)}, header = false] {figs/cyc_block_error.dat};
\pgfplotsextra{\yafdrawaxis{0}{0.5}{0}{0.6}}
\end{axis}
\end{tikzpicture}
} 
\subfloat[\aobmfsym]{%
\begin{tikzpicture}
\begin{axis}
\addplot table[x index = 0, y expr = {\thisrowno{1} / (95 * 95)}, header = false] {figs/sym_block_error.dat};
\addplot table[x index = 0, y expr = {\thisrowno{2} / (95 * 95)}, header = false] {figs/sym_block_error.dat};
\pgfplotsextra{\yafdrawaxis{0}{0.5}{0}{0.6}}
\end{axis}
\end{tikzpicture}
}
\subfloat[\acobmfsym]{%
\begin{tikzpicture}
\begin{axis}
\addplot table[x index = 0, y expr = {\thisrowno{1} / (95 * 95)}, header = false] {figs/sym_cyc_block_error.dat};
\addplot table[x index = 0, y expr = {\thisrowno{2} / (95 * 95)}, header = false] {figs/sym_cyc_block_error.dat};
\pgfplotsextra{\yafdrawaxis{0}{0.5}{0}{0.6}}
\end{axis}
\end{tikzpicture}
}
\end{center}
\caption{Error as a function of noise. Here the error is the proportion of disagreements between the reconstructed matrix
and either the noise-free or the noisy matrix. The decomposition was done using the noisy matrix.}
\label{fig:synth_data}
\end{figure*}

With lower leves of noise (\SI{35}{\percent} for \aobmf and \acobmf and \SI{25}{\percent} for the symmetric variants), the reconstruction of the original data is more accurate. With higher levels of noise, the noise has destroyed so much of the structure that the algorithms start fitting to the noise only, with a clear reduction of the quality versus the original data. 

It is also worth noticing that \aobmf obtains exact decompositions when the data has no noise; the other methods introduce a slight error even in these cases emphasizing their more complex setting.

\subsection{Scalability}
\label{sec:exp:scalability}

In this section we test how well \aobmf scales to larger data sets and how well it benefits from multiple cores. These experiments were executed on a server with 40 cores of Intel Xeon E7-4870 processors running at \SI{2.4}{\giga\hertz}. The algorithm was compiled using GCC 8.1.0 and the parallel code uses the OpenMP library.

To test the scalability, we generated \by{n}{n} square matrices with $n=2^i$ for $i=9,\ldots, 13$. All matrices have a density of approximately \SI{24}{\percent}. The results are presented in Figure~\ref{fig:scalability:size}. 

\begin{figure*}
  \begin{center}
    \subfloat[Scalability w.r.t. size\label{fig:scalability:size}]{%
      \begin{tikzpicture}
\begin{axis}[xlabel={size},ylabel={time (sec)},
    width = 6cm,
    height = 2.5cm,
    cycle list name=yaf,
    scale only axis,
    tick scale binop=\times,
    x tick label style = {/pgf/number format/set thousands separator = {\,}},
    y tick label style = {/pgf/number format/set thousands separator = {\,}},
    scaled ticks = false,
    xmax = 8192,
    ytick = {0, 200, ..., 1400},
    xtick = {512, 1024, 2048, 4096, 8192},
    xticklabels = {$2^{9}$, $2^{10}$, $2^{11}$, $2^{12}$, $2^{13}$},
    pin distance = 3mm,
    legend pos = north west,
    no markers
    ]
\addplot table[x index = 0, y index = 1, header = false] {figs/scaling.dat};
\pgfplotsextra{\yafdrawaxis{0}{0.5}{0}{1330}}
\end{axis}
\end{tikzpicture}

%
    }%
    \subfloat[Scalability w.r.t. number of cores\label{fig:scalability:cores}]{%
      \begin{tikzpicture}
\begin{axis}[xlabel={cores},ylabel={time (sec)},
    width = 6cm,
    height = 2.5cm,
    cycle list name=yaf,
    scale only axis,
    tick scale binop=\times,
    x tick label style = {/pgf/number format/set thousands separator = {\,}},
    y tick label style = {/pgf/number format/set thousands separator = {\,}},
    scaled ticks = false,
    xmax = 32,
    ytick = {0, 4, ..., 28},
    xtick = {1, 2, 4, 8, 16, 32},
    pin distance = 3mm,
    legend pos = north west,
    no markers
    ]
\addplot table[x index = 0, y index = 1, header = false] {figs/cores.dat};
\pgfplotsextra{\yafdrawaxis{0}{0.5}{0}{27}}
\end{axis}
\end{tikzpicture}

%
    }%
  \end{center}
  \caption{Scalability with respect to the size and number of cores.}
  \label{fig:scalability}
\end{figure*}
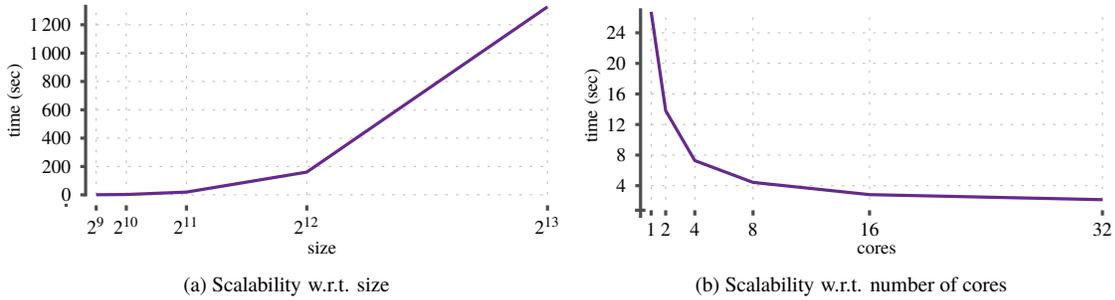

The algorithm shows very good scalability over the full range, although it does get slower when the data size increases from $2^{12}$ to $2^{13}$. It should be noted, though, that as the density is constant, the number of non-zeros in the matrices increases as the square of the matrix size. Hence, \aobmf exhibits linear growth with respect to the number of non-zero elements. 

Algorithm~\ref{alg:greedy} is almost embarrassingly parallel over the different seeds vectors. Hence, we parallellized the test of different seeds, and tested how the algorithm behaves with increased number of cores. The results are in Figure~\ref{fig:scalability:cores}, where we can see that the speed-up is essentially linear up to $4$ cores, slightly slower until $16$ cores, and only marginal gains are available when increasing the number of cores to $32$, indicating that at the algorithm has become memory bus constrained.

Overall, the experiments show that the algorithm scales very well, and is able to benefit from modern multi-core computers. We study further speed-up options later in Section~\ref{sec:exp:real:seeds}.

\subsection{Experiments with Real-World Data}
\label{sec:exp:real}

We now turn to real-world data sets. We used six different real-world data sets, selected to offer a wide variety of different types of data. The data sets we used are as follows. \lesmis is a standard benchmark data\footnote{\url{http://moreno.ss.uci.edu/data.html}} of the characters of Victor Hugo's novel \textit{Les Mis\'erables}. \now is a palaeontological data\footnote{NOW 030717, \url{http://www.helsinki.fi/science/now/}} in the form of a locations-by-genera matrix, giving information where different fossiles have been found. \news is a subset of the famous 20Newsgroups data\footnote{\url{http://qwone.com/~jason/20Newsgroups/}} consisting four newsgroups and \num{100} terms. \newssym the terms-by-terms co-occurrence matrix based on \news. \locations is locations-by-locations matrix indicating mammal species co-location in the northern hemisphere: the data has a $1$ in element $(i,j)$ if locations $i$ and $j$ have at least five mammals in common. The data is based on the IUC Red List data.\!\footnote{\url{http://www.iucnredlist.org/technical-documents/spatial-data}} The final data set, \mammals, contains a species-by-species co-inhabitation matrix.\!\footnote{Available for research purposes from the Societas Europaea Mammalogica at \url{http://www.european-mammals.org}} The data set properties are summarized in Table~\ref{tab:properties}.

\begin{table}[tb]
  \centering
  \caption{Properties of real-world data sets. Rank indicates the rank used in the decomposition.}
  \label{tab:properties}
  \small
  \begin{tabular}{@{}l
    S[table-figures-integer=4,table-figures-decimal=0,table-number-alignment=right]
    S[table-figures-integer=4,table-figures-decimal=0,table-number-alignment=right]
    S[table-figures-integer=2,table-figures-decimal=2,table-number-alignment=right]
    c
    S[table-figures-integer=2,table-figures-decimal=0,table-number-alignment=right]
    @{}}
    \toprule
    data      & {rows} & {cols} & {\% of 1s} & sym. & rank \\
    \midrule
    \lesmis & 77 & 77 & 8.57 & Yes & 10 \\
    \now    & 124 & 139 & 11.48 & No & 10 \\
    \news   & 100 & 348 & 6.30 & No & 10 \\
    \newssym & 100 & 100 & 48.54 & Yes & 10 \\
    \locations & 3203 & 3203 & 8.42 & Yes & 50 \\
    \mammals & 194 & 194 & 58.04 & Yes & 10 \\
    \bottomrule
  \end{tabular}
\end{table}

To the best of our knowledge, this is the first work to address the ordered Boolean matrix factorization problem. To understand what kind of an effect the ordering constraint has to the reconstruction error, we compare our results with those of \asso \cite{miettinen08discrete}. The \asso algorithm is a well-known method for computing the standard Boolean matrix factorization. We used an implementation available from the author\footnote{\url{https://cs.uef.fi/~pauli/basso/basso-0.5.tar.gz}} and set the rank for \asso the same as for our algorithms, and used threshold values $\tau=\{0.2, 0.4, 0.6, 0.8\}$.

For symmetric data sets, we also computed the symmetric Boolean factorization. This was done by first computing the standard $\mX^T\bprod\mY$ factorization, and then testing whether $\mX^T\bprod\mX$ or $\mY^T\bprod\mY$ gives smaller reconstruction error and using that one. This version of \asso is denoted \assosym.

\subsubsection{Reconstruction errors}
\label{sec:exp:real:errors}

We first compute the reconstruction errors for the various data sets. To facilitate the comparisons, we report the \emph{relative} reconstruction error
\[
  \frac{\norm*{\mD - \mX^T\bprod\mY}_F^2}{\norm{\mD}_F^2}\;.
\]
The results of all datasets are given in Table~\ref{tab:real:rel:asym}. 

\sisetup{
  table-figures-integer=1,
  table-figures-decimal=4,
  round-mode=places,
  round-precision=2
}
\begin{table*}
  \centering
  \caption{Relative errors with asymmetric (left) and symmetric (right) algorithms on real-world data.}
  \label{tab:real:rel:asym}
  \setlength{\tabcolsep}{3pt}
  \small
  \begin{tabular}{@{}l*{6}{S}@{}}
    \toprule
    & {\lesmisa} & {\now} & {\newsa} & {\newssym} & {\locations} & {\mammals} \\
    \midrule
    \aobmf & 0.3562 & .7123 & .7428 & .3166 & .4026 & .2569 \\
    \acobmf & 0.3562 & .7194 & .7428 & .3199 & .4027 & .2556 \\
    \asso & 0.3287 & .7087 & .7232 & .2865 & .3444 & .2617 \\
    \bottomrule
  \end{tabular}

  \begin{tabular}{@{}l*{4}{S}@{}}
    \toprule
    & {\lesmisa} & {\newssym} & {\locations} & {\mammals} \\
    \midrule
    \aobmfsym & .4035 & .3463 & .4762 & .2670 \\
    \acobmfsym & .4055 & .3638 & .4525 & .2716 \\
    \assosym & .6614 & .3339 & 1.0183 & .3317 \\
    \bottomrule
  \end{tabular}

\end{table*}


In case of asymmetric decompositions, \asso is -- as expected, as its factor matrices are not restricted to unimodal or cyclic -- almost always slightly better than either \aobmf or \acobmf. This difference is, however, very small in many data sets (only \SI{8}{\percent} in \lesmis and \SI{0.5}{\percent} in \now). A remarkable exception is the \mammals data, where \asso is in fact worse than either \aobmf or \acobmf. As the data set is the densest of the ones we tested, it is possible that \asso was unable to obtain good candidates from it with the rounding thresholds we tried.

There is almost no difference between \aobmf and \acobmf in the terms of reconstruction error in these data sets. Usually, \aobmf is on par or slightly better than \acobmf, except again in \mammals, where \acobmf is slightly better. The asymmetric data sets, \now and \news, cause the highest reconstruction errors at over \SI{70}{\percent}. It should be noted, though, that also \asso has similarly high errors with these data sets, indicating that they might not have strong Boolean low-rank structure.

In symmetric decompositions, the relationship between the ordered BMF algorithms and \asso is reversed, with \assosym being often the worse method (with the exception of \newssym). This is not very surprising, given that \asso is not designed for symmetric decompositions. The errors are slightly worse than with the asymmetric algorithms, highlighting the complexity of finding the symmetric decompositions. 

\subsubsection{Changing the seeds}
\label{sec:exp:real:seeds}

In the above experiments, we used the columns as the seeds $S$ for the algorithm (cf. Algorithm~\ref{alg:greedy}). This slows the algorithm down, as it has to attempt all of the potential seeds. In this section we study if we can improve the running time without hurting the reconstruction error by sampling only some of the columns for the seed set $S$.

In particular, we sampled \SI{10}{\percent} of the columns uniformly at random to create the seed set. As the algorithm scales linearly with the number of seeds, this provides an order of magnitude speed-up. To test the quality, we repeated the sampling ten times and report the average relative reconstruction errors and standard deviations in Table~\ref{tab:real:seed_rnd:rel:asym}.

\sisetup{
  separate-uncertainty=true,
  table-figures-integer=1,
  table-figures-decimal=2,
  table-number-alignment=center,
}
\begin{table*}
  \centering
  \caption{Average relative errors and standard deviation with random columns as seeds for asymmetric algorithms on real-world data. Ten random samples.}
  \label{tab:real:seed_rnd:rel:asym}
  \small
    \begin{tabular}{@{}l*{6}{S@{$\,\pm\,$}S}@{}}
    \toprule
    & \multicolumn{2}{l}{\lesmis} & \multicolumn{2}{l}{\now} & \multicolumn{2}{l}{\news} & \multicolumn{2}{l}{\newssym} & \multicolumn{2}{l}{\locations} & \multicolumn{2}{l}{\mammals} \\
    \midrule
    \aobmf & 0.5147 & 0.0836 & 0.7562 & 0.0173 & 0.8293 & 0.0223 & 0.3203 & .0 & 0.5301 & .0 & 0.2573 & .0 \\
    \acobmf & 0.4639 & 0.0489 & 0.7555 & 0.0223 & 0.8277 & 0.0200 & 0.3245 & 0.01 & 0.5149 & .0 & 0.2571 & .0 \\
    \bottomrule
  \end{tabular}

\end{table*}

The first thing to notice in Table~\ref{tab:real:seed_rnd:rel:asym} are the low standard deviations; less than \SI{3}{\percent} in almost all data sets. The reconstruction errors are also only slightly higher than those in Table~\ref{tab:real:rel:asym}; for instance, \aobmf with \now has only \SI{6}{\percent} higher error on average when using random sampling. In most cases the speed-up obtained by the sampling is significant compared to the loss in accuracy.

\subsection{Visualizing the Graphs}
\label{sec:exp:visual}

One of the motivations for the ordered BMF is that it allows the convenient visualization of the graphs using edge bundles (or ribbons) between nodes that are placed in a circle. In this section we explore some of these visualizations and explain what we can learn from the respective data sets using them. In the following plots, the edge bundles and the ordering are obtained form the factorization. Further visualizations can be found in Appendix~\ref{sec:appendix:visualizations}.

\textbf{The \lesmis data:}
The visualization of the \lesmis data is presented in Figure~\ref{fig:lesmis}. Most edge bundles form a circular segment indicating that all of the nodes under the segment are connected to each other (the characters appear in the same parts of the book). Some of the bundles are contained in other bundles, indicating important subset of characters. Multiple bundles intersect on a node at south-east of the circle called \emph{Valjean} -- the protagonist of the book. 

\begin{figure}[ht!]
  \centering
  \begin{tikzpicture}[scale=0.8]
    \input{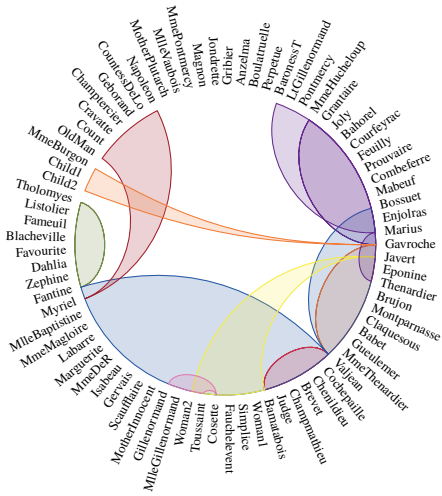}
  \end{tikzpicture}
  \caption{Visualization of the \lesmis data with the ribbons and ordering from \acobmf.}
  \label{fig:lesmis}
\end{figure}

\textbf{The \mammals data:}
The second data set is the \mammals data, in Figure~\ref{fig:mammals}. For a clearer visualization, we only consider \num{134} species that do not appear too frequently in the data, as such species are neighbours of every other species in graph. The edge bundles in Figure~\ref{fig:mammals} are essentially rotating around the middle. This probably corresponds to the change of fauna when moving from north to south. The change is gradual, hence two consecutive edge bundles have a significant overlap, but over longer distance, the change in the fauna becomes more obvious and the edge bundles are more disjoint.  This gives a good intuition about the structure of the  data.

\begin{figure}[ht!]
  \centering
  \begin{tikzpicture}[scale=0.8]
    \input{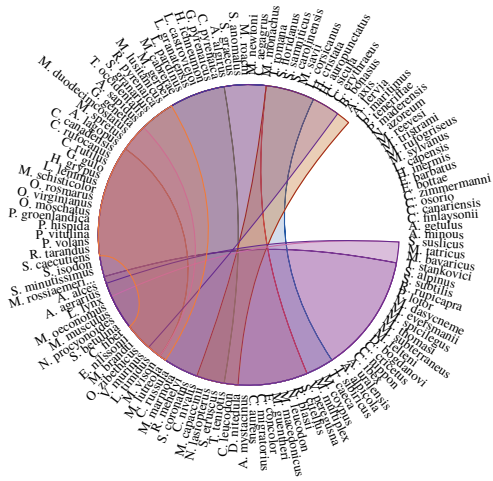}
  \end{tikzpicture}
  \caption{Visualization of the \mammals data with the ribbons and ordering from \aobmf.}
  \label{fig:mammals}
\end{figure}


\section{Related Work}
\label{sec:related-work}

\emph{Boolean matrix factorization} (BMF) has received increasing interest in the data analysis community~\cite{miettinen08discrete,lucchese13unifying,belohlavek10discovery,rukat17bayesian,ravanbakhsh16boolean,maurus16ternary,hess17c-salt,hess17primping,miettinen14mdl4bmf,karaev15getting}, proving to be a versatile tool for analyzing Boolean matrices. Many different algorithms have been proposed, including algorithms based on candidate creation and selection \cite{miettinen08discrete,lucchese13unifying}, proximal alternations \cite{hess17primping}, and message passing \cite{ravanbakhsh16boolean}, to name but a few. It has also found applications in diverse fields, such as bioinformatics~\cite{corrado14ptrcombiner}, information extraction~\cite{cergani13discovering}, and lifted inference~\cite{broeck13complexity}. To the best of our knowledge, however, the ordering constraint is not studied in earlier work related to Boolean matrix factorization.

\emph{Tiling} databases~\cite{geerts04tiling} can be seen as a restricted
version of BMF, where the factorization cannot express any $0$s as $1$.
\emph{Geometric tiling} \cite{gionis04geometric}  is a variation thereof, where
the tiles have to be consecutive. The main difference to our work is a
different optimization function, \cite{gionis04geometric} uses log-likelihood,
and that it assumes that the order is already given, for example, by spectral ordering,
whereas we discover the order on the fly.

A binary matrix has the \emph{consecutive ones property} (C1P) if its columns can be permuted so that all rows have all 1s consecutively. The pq-trees can be used to check for the C1P~\cite{booth1976testing} and \citet{atkins98spectral} propose spectral ordering algorithm. The spectral ordering approach is used in \cite{gionis04geometric} to permute the data for finding the geometric tiles.


\section{Conclusions}
\label{sec:conclusions}

Ordered Boolean matrix factorization (\obmf) and its variations (\cobmf, \obmfsym) are restricted versions of Boolean matrix factorization, requiring the factors to have the consecutive ones property (or be cyclic, in case of \cobmf). This restriction facilitates the interpretation of the factorization, in particular in the case of the edge bundle visualizations of graphs, as we saw in Section~\ref{sec:exp:visual}. On the other hand, the restriction yields higher reconstruction errors, though our experiments show that the difference to state-of-the-art Boolean matrix factorization algorithm is usually very small.

In this paper we laid the theoretical foundations of the \obmf problem and its variations, and proposed algorithms based on the pq-trees. An important part of the proposed algorithm is the choice of the seed vectors. In this paper, we mostly used all columns of the data as the seed, though the experiments in Section~\ref{sec:exp:real:seeds} show that sampling the columns could work equally well. An interesting question for the future is whether other methods for selecting the seeds would yield better reconstruction errors.

In the problem setting of this paper, the user provides the rank of the decomposition and the goal is to minimize the reconstruction error over the rank-$k$ \obmf decompositions. A common variant in the Boolean matrix factorization world is to make the rank a free variable and replace the target function with measure that penalizes for higher ranks (see, e.g.~\cite{miettinen14mdl4bmf,lucchese13unifying,hess17primping}). The \emph{Minimum Description Length} principle is a common approach. The ordered nature of our factor matrices could help with finding more efficient MDL decompositions, as the factor matrices are easier to compress using run-length encoding or similar approaches.


\bibliographystyle{abbrvnat}
\bibliography{library}  

\begin{thebibliography}{18}
\providecommand{\natexlab}[1]{#1}
\providecommand{\url}[1]{\texttt{#1}}
\expandafter\ifx\csname urlstyle\endcsname\relax
  \providecommand{\doi}[1]{doi: #1}\else
  \providecommand{\doi}{doi: \begingroup \urlstyle{rm}\Url}\fi

\bibitem[Atkins et~al.(1998)Atkins, Boman, and Hendrickson]{atkins98spectral}
J.~E. Atkins, E.~G. Boman, and B.~Hendrickson.
\newblock {A Spectral Algorithm for Seriation and the Consecutive Ones
  Problem}.
\newblock \emph{SIAM J. Comput.}, 28\penalty0 (1):\penalty0 297--310, 1998.

\bibitem[B{\v e}lohl{\'a}vek and Vychodil(2010)]{belohlavek10discovery}
R.~B{\v e}lohl{\'a}vek and V.~Vychodil.
\newblock {Discovery of optimal factors in binary data via a novel method of
  matrix decomposition}.
\newblock \emph{J. Comput. Syst. Sci.}, 76\penalty0 (1):\penalty0 3--20, 2010.

\bibitem[Booth and Lueker(1976)]{booth1976testing}
K.~S. Booth and G.~S. Lueker.
\newblock Testing for the consecutive ones property, interval graphs, and graph
  planarity using pq-tree algorithms.
\newblock \emph{J. Comput. Syst. Sci.}, 13\penalty0 (3):\penalty0 335--379,
  1976.

\bibitem[Cergani and Miettinen(2013)]{cergani13discovering}
E.~Cergani and P.~Miettinen.
\newblock {Discovering relations using matrix factorization methods}.
\newblock In \emph{CIKM '13}, pages 1549--1552, 2013.

\bibitem[Corrado et~al.(2014)Corrado, Tebaldi, Bertamini, Costa, Quattrone,
  Viero, and Passerini]{corrado14ptrcombiner}
G.~Corrado, T.~Tebaldi, G.~Bertamini, F.~Costa, A.~Quattrone, G.~Viero, and
  A.~Passerini.
\newblock {PTRcombiner: mining combinatorial regulation of gene expression from
  post-transcriptional interaction maps}.
\newblock \emph{BMC Genomics}, 15\penalty0 (1), Apr. 2014.

\bibitem[Geerts et~al.(2004)Geerts, Goethals, and
  Mielik{\"a}inen]{geerts04tiling}
F.~Geerts, B.~Goethals, and T.~Mielik{\"a}inen.
\newblock {Tiling databases}.
\newblock In \emph{DS '04}, pages 278--289, 2004.

\bibitem[Gillis and Vavasis(2015)]{gillis15complexity:arxiv}
N.~Gillis and S.~A. Vavasis.
\newblock {On the Complexity of Robust PCA and $\ell_1$-norm Low-Rank Matrix
  Approximation}.
\newblock \emph{arXiv}, 2015.

\bibitem[Gionis et~al.(2004)Gionis, Mannila, and
  Sepp{\"a}nen]{gionis04geometric}
A.~Gionis, H.~Mannila, and J.~K. Sepp{\"a}nen.
\newblock {Geometric and Combinatorial Tiles in 0{\textendash}1 Data}.
\newblock In \emph{PKDD '04}, pages 173--184, 2004.

\bibitem[Hess and Morik(2017)]{hess17c-salt}
S.~Hess and K.~Morik.
\newblock {C-SALT: Mining Class-Specific ALTerations in Boolean Matrix
  Factorization}.
\newblock In \emph{ECMLPKDD '17}, pages 547--563, 2017.

\bibitem[Hess et~al.(2017)Hess, Morik, and Piatkowski]{hess17primping}
S.~Hess, K.~Morik, and N.~Piatkowski.
\newblock {The PRIMPING routine{\textemdash}Tiling through proximal alternating
  linearized minimization}.
\newblock \emph{Data Min. Knowl. Discov.}, 31\penalty0 (4):\penalty0
  1090--1131, May 2017.

\bibitem[Karaev et~al.(2015)Karaev, Miettinen, and Vreeken]{karaev15getting}
S.~Karaev, P.~Miettinen, and J.~Vreeken.
\newblock {Getting to Know the Unknown Unknowns: Destructive-Noise Resistant
  Boolean Matrix Factorization}.
\newblock In \emph{SDM '15}, pages 325--333, 2015.

\bibitem[Lucchese et~al.(2013)Lucchese, Orlando, and
  Perego]{lucchese13unifying}
C.~Lucchese, S.~Orlando, and R.~Perego.
\newblock {A Unifying Framework for Mining Approximate Top-k Binary Patterns}.
\newblock \emph{IEEE Trans. Knowl. Data Eng.}, 26\penalty0 (12):\penalty0
  2900--2913, Dec. 2013.

\bibitem[Maurus and Plant(2016)]{maurus16ternary}
S.~Maurus and C.~Plant.
\newblock {Ternary Matrix Factorization: problem definitions and algorithms}.
\newblock \emph{Knowl. Inf. Syst.}, 46\penalty0 (1):\penalty0 1--31, Jan. 2016.

\bibitem[Miettinen and Vreeken(2014)]{miettinen14mdl4bmf}
P.~Miettinen and J.~Vreeken.
\newblock {MDL4BMF: Minimum Description Length for Boolean Matrix
  Factorization}.
\newblock \emph{ACM Trans. Knowl. Discov. Data}, 8\penalty0 (4):\penalty0 --31,
  Oct. 2014.

\bibitem[Miettinen et~al.(2008)Miettinen, Mielik{\"a}inen, Gionis, Das, and
  Mannila]{miettinen08discrete}
P.~Miettinen, T.~Mielik{\"a}inen, A.~Gionis, G.~Das, and H.~Mannila.
\newblock {The Discrete Basis Problem}.
\newblock \emph{IEEE Trans. Knowl. Data Eng.}, 20\penalty0 (10):\penalty0
  1348--1362, Oct. 2008.

\bibitem[Ravanbakhsh et~al.(2016)Ravanbakhsh, P{\'o}czos, and
  Greiner]{ravanbakhsh16boolean}
S.~Ravanbakhsh, B.~P{\'o}czos, and R.~Greiner.
\newblock {Boolean Matrix Factorization and Noisy Completion via Message
  Passing}.
\newblock In \emph{ICML '16}, 2016.

\bibitem[Rukat et~al.(2017)Rukat, Holmes, Titsias, and Yau]{rukat17bayesian}
T.~Rukat, C.~C. Holmes, M.~K. Titsias, and C.~Yau.
\newblock {Bayesian Boolean Matrix Factorisation}.
\newblock In \emph{ICML '17}, pages 2969--2978, July 2017.

\bibitem[van~den Broeck and Darwiche(2013)]{broeck13complexity}
G.~van~den Broeck and A.~Darwiche.
\newblock {On the Complexity and Approximation of Binary Evidence in Lifted
  Inference}.
\newblock In \emph{NIPS '13}, pages 2868--2876, 2013.

\end{thebibliography}


\appendix
\section{Proofs} 
\label{sec:appendix:proofs}

\begin{proof}[Proof of Theorem~\ref{thm:obmf:np_hard}]
  In this case, we are looking for a decomposition of format $\mD \approx \vx^T \vy$, where $\mD\in\B^{n\times m}$, $\vx\in\B^n$, and $\vy\in\B^m$. Notice that (\textit{i}) whether we use normal or Boolean algebra does not matter in this case; and (\textit{ii}) we can always find the ordering after we have found the decomposition, as we only need to order the vectors $\vx$ and $\vy$. But this problem, the \emph{rank-1 binary matrix factorization} problem, is known to be \NP-hard~\citep{gillis15complexity:arxiv}, finalizing the proof.
\end{proof}

\begin{proof}[Proof of Theorem~\ref{thm:obmf:inapprox}]
The decision problem is obviously in \NP.

We prove the hardness by reduction from \hamilton, where
we are given a graph $G = (V, E)$ and asked whether there is a hamiltonian
path, that is, a path visiting every vertex exactly once.

Assume that we are given a graph $G = (V, E)$ with $n$ vertices and $m$ edges.
Assume that we have some arbitrary order on the vertices $V = v_1, \ldots,
v_n$, and on the edges $E = e_1, \ldots, e_m$.

Let us define $\mD$ first. The dataset will be of size $\by{(n + m + 1)}{(3m +
1)}$.  To define the matrix, we split the rows in two parts $R = r_1, \ldots,
r_n$ and $S = s_0, \ldots, s_m$, containing respectively $n$ and $m$ rows.
Similarly, we split the columns in 3 parts, $X = x_1, \ldots, x_m$, $Y = y_1, \ldots, y_m$,
$Z = y_0, \ldots, y_m$.

The 1s in $\mD$ are as follows. for each edge $e_\ell = (v_i, v_j)$, we set
the cells
$(r_i, x_\ell)$
$(r_j, x_\ell)$
$(r_i, y_\ell)$
$(r_j, y_\ell)$ to be 1.
For two adjacent edges $e_\ell$ and $e_{\ell + 1}$, we set the cells
$(s_\ell, y_\ell)$
$(s_\ell, z_\ell)$
$(s_\ell, x_{\ell + 1})$.
Finally, we set $(s_0, x_1)$, $(s_0, z_0)$, and
$(s_m, y_m)$, $(s_m, z_m)$ to be 1.
The remaining values are 0.

We argue that there is a zero-error solution for \obmf using $k = 3m - n + 2$ if and only
there is a hamiltonian path.

Let us prove the easy direction: assume that there is a hamiltonian path.
To that end, let us permute the rows and columns $\mD$ such that the factor matrices
do not have gap zeros.
Permute $\mD$ as follows: Set the column order as $z_0, x_1, y_1, z_1, x_2, y_2, \ldots$.
Order the rows in $R$ according to the hamiltonian path, followed by the rows in $S$.
We denote the resulting matrix by $\mD'$.
There is a zero-solution if the ones in $\mD'$ are a union of $k$ contiguous blocks.
The $k$ blocks are as follows: $m + 1$ blocks covering individual rows in $S$,
$n - 1$ blocks covering edges along the hamiltonian path (this can be done
since the corresponding rows in $R$ and the corresponsding columns in $X$ and $Y$ are adjacent), 
and $2(m - n + 1)$ blocks to cover the remaining edges, 2 blocks per edge. 
This covers all 1s using $m + 1 + n - 1 + 2(m - n + 1) = k$ blocks.

Let us prove the other direction.
Assume that there is zero-error solution, and let $\mD'$ be the permuted version of $\mD$
with no gap zeros.
Then the ones in $\mD'$ must be a union of $k$ contiguous blocks. 
For a column index $i$, we define $f_i$ to be the number of blocks started at
the $i$th column.
Let us also define $g_i$ to be the number of blocks ended at $i$th columns.
Trivially, $\sum_i f_i + g_i = 2k$.

We say that an edge $(v_i, v_j) \in E$ is \emph{active} if $i$ and $j$ are adjacent
in $\mD'$. Let $h$ be the total number of active edges. Note that we have
$h \leq n - 1$.  Assume for a moment that $h = n - 1$
and let $w_1, \ldots w_n$ be the vertices ordered according to the order of $R$
in $\mD'$. Since $h = n - 1$, we are forced to have $(w_i, w_{i + 1}) \in E$. 
This implies that $w_1, \ldots, w_n$ is a hamiltonian path.

We will now argue that $h \geq n - 1$.


Consider two adjacent columns at $i$ and $i + 1$. If none of the columns
are in $Z$, then both columns contain 1 that is not in the other column.
This forces $g_i + f_{i + 1} \geq 2$. The same argument holds if both
columns are in $Z$.

Assume that the $j$th column is in $X$ and $(j + 1)$th column is in $Z$.
Assume that $g_i + f_{i + 1} = 1$. Let $a$ and $b$ be the rows in $R$
that are active in the $j$th columns.
Since $Z$ does not have active rows, the block(s) covering $a$ and $b$
must terminate, and since $g_i \geq 1$, we have only block, implying
that $a$ and $b$ are adjacent. The same result holds if we replace $X$
with $Y$ or permute the order of the two columns. 
To summarize, if $g_i + f_{i + 1} = 1$, then either $i$th or the $(i + 1)$th
column corresponds to an active edge.

In addition, we must have $f_1 \geq 1$ and $g_{3m + 1} \geq 1$ as these columns have 1s.
This leads to
\[
\begin{split}
	2(3m - n + 2) & = 2k = \sum_{i = 1}^{3m + 1} f_i + g_i \\
	& = f_1 + g_{3m + 1} + \sum_{i = 1}^{3m} f_{i + 1} + g_i \\
	& \geq 2(3m + 1) - 2h,
\end{split}
\]
proving the result.
\end{proof}

\begin{proof}[Proof of Lemma~\ref{lem:borderq}]
Let $S$ be the optimal border-compatible set. Then there is $i$
such that $S$ is a union of the best border-compable set of $c_i$
and either the union of all leaves in $c_1, \ldots, c_{i - 1}$
or $c_{i + 1}, \ldots, c_\ell$.
\end{proof}

\begin{proof}[Proof of Lemma~\ref{lem:innerq}]
Let $S$ be the optimal compatible set. Then $S$ is either included
completely within one child, or
there are indices $i < j$
such that $S$ is a union of the best border-compable sets of $c_i$, $c_j$,
and the union of all leaves in $c_{i + 1}, \ldots, c_{j - 1}$.
\end{proof}

\begin{proof}[Proof of Lemma~\ref{lem:borderp}]
Let $S$ be the optimal border-compatible set. Then there is $i$ such that $S$
is a union of the best border-compable set of $c_i$ and the union of all leaves
of some children.

Let $w$ be a child of $v$, if $\total{w} \geq 0$, then having the leaves of
$w$ in $S$ has positive gain. Let $P$ be these children.  The total gain
corresponds of having these children is $\sum_i \max(\total{v}, 0)$.

We need to transform one of the children to a partial. Let $w$ be a child of
$v$.  If $\total{w} < 0$, then $v \notin P$ and adding $w$ will have a gain of
$\border{w}$. If $\total{w} \geq 0$, then $v \in P$, and transforming $w$
from a fully-covered node to a partial node will have a gain of $\border{w} - \total{w}$.
In summary, the gain is equal to $g(w)$. Thus, selecting the vertex with
the maximal $g(w)$ should be the partial child in $S$.
\end{proof}

\begin{proof}[Proof of Lemma~\ref{lem:innerp}]
Let $S$ be the optimal compatible set.  Then $S$ is either included completely
within one child, or
$S$ is a union of some children and possibly up to two
of the best border-compable sets for some $c_i$ and $c_j$.

Let $w$ be a child of $v$, if $\total{w} \geq 0$, then having the leaves of
$w$ in $S$ has positive gain. Let $P$ be these children.  The total gain
corresponds of having these children is $\sum_i \max(\total{v}, 0)$.

As shown in the proof of Lemma~\ref{lem:borderp}, $b_1$ and $b_2$ correspond
the top-2 border-compatible sets. It may happen that $b_1$ or $b_2$ are negative,
in which case we simply do not add them to $S$. Thus the total gain of
border-compatible sets is $\max(b_1, 0) + \max(b_2, 0)$.
\end{proof}

\section{Further Visualizations}
\label{sec:appendix:visualizations}

Here we present for the \newssym and \locations data sets.

\paragraph{The \newssym data}
\label{sec:exp:visual:terms}

The visualization of the \newssym data, in Figure~\ref{fig:4news}, is markedly different from Figure~\ref{fig:lesmis}. Here, most bundles overlap each other. This indicates that many of these terms are used together in different posts. Yet, we can also identify specialized groups of terms. At the left of Figure~\ref{fig:4news}, we have a blue bundle, from \emph{mission} to \emph{nasa}, that contains terms used when discussing space programs. This overlaps with a larger orange bundle, from \emph{chip} to \emph{tap}, containing terms related to cryptography. 

\begin{figure}
  \centering
  \begin{tikzpicture}
    \input{4news_graph}
  \end{tikzpicture}
  \caption{Visualization of the \newssym data with the ribbons and ordering from \aobmf.}
  \label{fig:4news}
\end{figure}

\paragraph{The \locations data}
\label{sec:exp:visual:locations}

For the \locations data, in Figure~\ref{fig:locations}, we cannot print any labels, as the data consists of \num{3203} geographical locations. For these results, we did a rank-$10$ decomposition. Most of the edge bundles again form segments along the edge of the circle, corresponding to locations with similar fauna. Few larger edge bundles cover most of these locations, as well, corresponding to more general biospheres. In this figure, many nodes have no edges drawn. This indicates that they were not part of any significant quasi-clique. 

\begin{figure}
  \centering
  \begin{tikzpicture}
    \fill[color=yafcolor1, opacity=0.2] (139.369:2.5) arc (139.369:186.912:2.5) to[in=319.369, out=366.912] (139.369:2.5) arc (139.369:186.912:2.5) to[in=319.369, out=366.912] (139.369:2.5);
\draw[color=yafcolor1] (139.369:2.5) arc (139.369:186.912:2.5) to[in=319.369, out=366.912] (139.369:2.5) arc (139.369:186.912:2.5) to[in=319.369, out=366.912] (139.369:2.5);
\fill[color=yafcolor2, opacity=0.2] (226.138:2.5) arc (226.138:273.344:2.5) to[in=406.138, out=453.344] (226.138:2.5) arc (226.138:273.344:2.5) to[in=406.138, out=453.344] (226.138:2.5);
\draw[color=yafcolor2] (226.138:2.5) arc (226.138:273.344:2.5) to[in=406.138, out=453.344] (226.138:2.5) arc (226.138:273.344:2.5) to[in=406.138, out=453.344] (226.138:2.5);
\fill[color=yafcolor3, opacity=0.2] (186.575:2.5) arc (186.575:226.138:2.5) to[in=366.575, out=406.138] (186.575:2.5) arc (186.575:226.138:2.5) to[in=366.575, out=406.138] (186.575:2.5);
\draw[color=yafcolor3] (186.575:2.5) arc (186.575:226.138:2.5) to[in=366.575, out=406.138] (186.575:2.5) arc (186.575:226.138:2.5) to[in=366.575, out=406.138] (186.575:2.5);
\fill[color=yafcolor4, opacity=0.2] (263.116:2.5) arc (263.116:292.451:2.5) to[in=443.116, out=472.451] (263.116:2.5) arc (263.116:292.451:2.5) to[in=443.116, out=472.451] (263.116:2.5);
\draw[color=yafcolor4] (263.116:2.5) arc (263.116:292.451:2.5) to[in=443.116, out=472.451] (263.116:2.5) arc (263.116:292.451:2.5) to[in=443.116, out=472.451] (263.116:2.5);
\fill[color=yafcolor5, opacity=0.2] (168.142:2.5) arc (168.142:302.904:2.5) to[in=348.142, out=482.904] (168.142:2.5) arc (168.142:302.904:2.5) to[in=348.142, out=482.904] (168.142:2.5);
\draw[color=yafcolor5] (168.142:2.5) arc (168.142:302.904:2.5) to[in=348.142, out=482.904] (168.142:2.5) arc (168.142:302.904:2.5) to[in=348.142, out=482.904] (168.142:2.5);
\fill[color=yafcolor6, opacity=0.2] (303.016:2.5) arc (303.016:321.224:2.5) to[in=483.016, out=501.224] (303.016:2.5) arc (303.016:321.224:2.5) to[in=483.016, out=501.224] (303.016:2.5);
\draw[color=yafcolor6] (303.016:2.5) arc (303.016:321.224:2.5) to[in=483.016, out=501.224] (303.016:2.5) arc (303.016:321.224:2.5) to[in=483.016, out=501.224] (303.016:2.5);
\fill[color=yafcolor7, opacity=0.2] (167.918:2.5) arc (167.918:331.452:2.5) to[in=366.687, out=511.452] (186.687:2.5) arc (186.687:199.725:2.5) to[in=347.918, out=379.725] (167.918:2.5);
\draw[color=yafcolor7] (167.918:2.5) arc (167.918:331.452:2.5) to[in=366.687, out=511.452] (186.687:2.5) arc (186.687:199.725:2.5) to[in=347.918, out=379.725] (167.918:2.5);
\fill[color=yafcolor8, opacity=0.2] (186.687:2.5) arc (186.687:199.725:2.5) to[in=347.918, out=379.725] (167.918:2.5) arc (167.918:331.452:2.5) to[in=366.687, out=511.452] (186.687:2.5);
\draw[color=yafcolor8] (186.687:2.5) arc (186.687:199.725:2.5) to[in=347.918, out=379.725] (167.918:2.5) arc (167.918:331.452:2.5) to[in=366.687, out=511.452] (186.687:2.5);
\fill[color=yafcolor1, opacity=0.2] (331.452:2.5) arc (331.452:346.513:2.5) to[in=511.564, out=526.513] (331.564:2.5) arc (331.564:346.513:2.5) to[in=511.452, out=526.513] (331.452:2.5);
\draw[color=yafcolor1] (331.452:2.5) arc (331.452:346.513:2.5) to[in=511.564, out=526.513] (331.564:2.5) arc (331.564:346.513:2.5) to[in=511.452, out=526.513] (331.452:2.5);
\fill[color=yafcolor2, opacity=0.2] (346.625:2.5) arc (346.625:359.888:2.5) to[in=407.15, out=539.888] (227.15:2.5) arc (227.15:273.231:2.5) to[in=526.625, out=453.231] (346.625:2.5);
\draw[color=yafcolor2] (346.625:2.5) arc (346.625:359.888:2.5) to[in=407.15, out=539.888] (227.15:2.5) arc (227.15:273.231:2.5) to[in=526.625, out=453.231] (346.625:2.5);
  \end{tikzpicture}
  \caption{Visualization of the \locations data with the ribbons and ordering from \aobmf.}
  \label{fig:locations}
\end{figure}


\end{document}